\documentclass[authoryear,a4paper, 12pt]{elsarticle}
\usepackage{natbib}

\usepackage{amsthm}
\usepackage{amssymb}
\usepackage{graphicx,amsmath}
\usepackage{tikz}
\usepackage{mathtools}
\usepackage{natbib}
\usepackage{booktabs,xcolor}
\usepackage{graphics}
\usepackage{enumitem}
\usepackage{subcaption}
\usepackage{appendix}
\usepackage{adjustbox}
\usetikzlibrary{patterns}
\usepackage{caption}
\usepackage{setspace}
\usepackage{comment}
\usepackage[colorlinks=true,linkcolor=blue, urlcolor=blue, citecolor=blue, breaklinks]{hyperref}
\journal{Journal of Mathematical Economics}
\usepackage{lscape}

\newtheorem{proposition}{Proposition}
\newtheorem{corollary}{Corollary}

\newtheorem{lemma}{Lemma}

\theoremstyle{definition}

\newtheorem{definition}{Definition}

\usepackage{etoolbox}
\makeatletter
\patchcmd{\ps@pprintTitle}
{Preprint submitted to}
{Forthcoming in the }
{}{}
\makeatother
\patchcmd{\emailauthor}{(#2)}{}{}{}
\patchcmd{\urlauthor}{(#2)}{}{}{}

\usepackage{tikz}
\usetikzlibrary{shapes}
\newcommand\encircle[1]{%
	\tikz[baseline=(X.base)] 
	\node (X) [draw, shape=circle, inner sep=0] {\strut #1};}

\newcommand\ensquare[1]{%
	\tikz[baseline=(X.base)] 
	\node (X) [draw, shape=diamond, inner sep=0cm] {\strut #1};}

\newcommand{\EE}{\mathbb{E}}

\newcommand{\MM}{\mathcal{M}}

\newcommand{\rk}{\text{rk}} 
\usepackage[justification=centering]{caption}
\usetikzlibrary{automata}
\renewcommand{\tablename}{Example} 

\begin{document}
	
	\begin{frontmatter}
		\title{On the Integration of Shapley--Scarf Markets
			\footnote{{\it Author contributions}: J. Ortega designed the research, conducted the simulations, and replied to reviewers. R. Kumar, K. Manocha and J. Ortega conducted the mathematical analysis. K. Manocha and J. Ortega wrote the paper.}
		}
		\author[add2]{Rajnish Kumar}
		\author[add1]{Kriti Manocha}
		\author[add2]{Josu\'e Ortega}
		\ead{j.ortega@qub.ac.uk (Corresponding author)}

		\address[add2]{Queen's Management School, Queen's University Belfast, UK.}		
		\address[add1]{Indian Statistical Institute, Delhi, India.}				
		\date{\today}
		
		\begin{abstract}
			We study the welfare consequences of merging Shapley--Scarf markets. Market integration can lead to large welfare losses and make the vast majority of agents worse-off, but is on average welfare-enhancing and makes all agents better off ex-ante. The number of agents harmed by integration is a minority when all markets are small or agents' preferences are highly correlated.
		\end{abstract}
		
	\end{frontmatter}
	
	\newpage
	\setcounter{footnote}{0}
	\onehalfspacing
	
	\section{Introduction}
	\label{sec:introduction}
	
	Shapley--Scarf markets, in which agents own one house each which they can exchange among themselves without using monetary transfers, have been helpful to analyse several real-life allocation problems, such as the assignment of campus housing to students \citep{chen2002improving}, house allocation with existing tenants \citep{abdulkadirouglu1999house} and kidney exchanges involving incompatible donor-patient pairs \citep{roth2004kidney}. A common complication in these allocation problems is that a big market is fragmented into several small and disjoint ones, causing inefficiencies. For example, house swaps in Australia are restricted to tenants within the same constituencies and community housing provider, blocking potentially beneficial exchanges \citep{powell2019construction}. Similarly, most kidney exchanges in the US are conducted locally, despite the existence of centralized clearinghouses, which if used could increase the number of transplants  by  up to  63  percent \citep{agarwal2019market}.
	
	Motivated by these observations, we investigate theoretically the welfare effects of integrating disjoint Shapley--Scarf markets. In our model, there are $k$ Shapley--Scarf markets with $n_j$ agents each ($n_j$ is potentially different for each market) and $n$ agents in total. The segregated allocation is obtained by treating each community separately and calculating the core allocation for each of them. The integrated allocation is the core allocation for the entire economy.
	
	Our first result (Proposition \ref{thm:worst-number}) states that up to, but not more than, $n-k$ agents may be harmed by integration, i.e., they receive a house they prefer more when trade is only allowed within their own disjoint markets. This upper bound holds for any choice of $n$ and $k$. It shows that Shapley--Scarf markets may fail to integrate because doing so could generate significantly more losers than winners. 
	
	Our second result (Proposition \ref{thm:worst-size}) concerns the size of the gains from integration in terms of house rank. For example, if an agent receives her 3rd best house before integration, but her 1st best after integration, the size of her gains from integration is $3-1=2$. Even if most agents are harmed by the merge of disjoint markets, integration may still be justified if the size of the gains from integration experienced by a few is substantially larger than the size of the losses from many. We show that, in the worst-case scenario, the size of the average gains from integration may be down to, but not less than, $\frac{-n^2+n+k^2+k}{2 n^2}$. This lower bound can be achieved for any choice of $n$ and $k$, and shows that, asymptotically, integration may increase the average house rank by 50\% of the size of agents' preference lists.
	
	Taken together, our first two results show that there are real obstacles to the integration of Shapley--Scarf markets. For example, if we have three small markets that merge into one with 60, 30 and 10 agents respectively, up to 97 agents may obtain a worse house after integration occurs, and on average (across all agents) each agent may receive a house 50 positions down on her preference list, equivalent to going from her top choice to her 51st choice.
	
	However, these results are obtained in worst-case scenarios, which occur only when preferences are very specific. Consequently, studying the expected gains from integration across all possible preference profiles may be more informative. Therefore, our third result studies the size of the expected gains from integration in random Shapley--Scarf markets, in which agents' preferences over houses are drawn uniformly and independently.
	
	In Proposition \ref{thm:avg-size}, we compute the exact expected gains from integration, which equal $\frac{(n+1)[ (n_j+1)H_{n_j}-n_j]}{n_j(n_j+1)n}-\frac{(n+1)H_n - n}{n^2}$ (where $H_n$ is the $n$-th harmonic number). This result shows that the expected welfare gains from integration are positive for all agents, and larger for agents belonging to smaller markets. Going back to our example of three markets integrating with 60, 30 and 10 agents, the agents of the market with size ten go up 16 positions in their expected house rank, whereas those in the market of size sixty also increase their expected allocated house rank, but only by 2 rank positions. Our third result gives some context to our first two propositions, and shows that on average we should expect an overall positive effect from integration in Shapley--Scarf markets for agents from all disjoint markets.
	
	Our fourth result (Proposition \ref{thm:avg-number}) establishes a connection between the number of trading cycles that occur in the top trading cycles algorithm and the expected number of agents harmed by integration. We use this connection to show that the expected number of agents harmed by integration in each economy is less than $n_j-\sqrt{2\pi n_j} - O(\log  n_j)$, and consequently the expected number of agents harmed by integration in the entire economy is smaller than $n-\sqrt{2\pi} (\sum_{j=1}^k \sqrt{n_j}) - O(\log \prod_{j=1}^k n_j)$. In our example regarding the integration of markets with size 60, 30 and 10, our result implies that the expected number of agents harmed by integration is less than 44, 19, and 4 for each respective market. A consequence of our result is that, when all markets are of the same size, the expected fraction of agents harmed by integration is less than 50\% whenever each market has less than $8 \pi\approx 25.13$ agents.
	
	A different approach to ensure that integration does not harm a majority of agents is to focus on specific preference domains. We find a preference domain that achieves this purpose, called sequential dual dictatorship, which enforces a particular correlation among agents' preferences. When preferences satisfy this property, we can guarantee that no more than 50\% of agents in any individual market are harmed by integration (Proposition \ref{thm:domain}). The sequential dual dictator property is equivalent to assigning the title of dictator to at most two agents at each step of the top trading cycle algorithm, therefore bounding the length of cycles that can occur.

	\paragraph{Structure of the paper} Section \ref{sec:lit} discusses the literature. 
	Section \ref{sec:model} presents our model. 
	Section \ref{sec:example} introduces a running example.
	Section \ref{sec:worst} presents worst-case results. 
	Section \ref{sec:random} discusses average-case results. 
	Section \ref{sec:domains} studies preference domains. 
	Section \ref{sec:concl} concludes.
	
	\section{Related Literature}
	\label{sec:lit}
	
	A few other papers study the effects of integration on variations of Shapley--Scarf markets. For example, \cite{ashlagi2014free} study the incentives for hospitals to fully reveal their patient–donor pairs to a centralized clearinghouse. In their model, agents do not have preferences but only dichotomous compatibility restrictions. Thus, welfare is measured by the size of the matching. They obtain worst- and average-case results that have a similar flavour to ours: the average-case cost for hospitals to fully integrate into a centralized clearinghouse is small, but the worst-case cost is high. In the same framework as them, \cite{toulis2015} propose a mechanism that is efficient and asymptotically individually rational for hospitals. Our paper differs from the aforementioned articles in that we measure welfare in terms of how desirable the integrated allocation is with respect to the segregated one, rather than by the number of total exchanges (which is constant in the canonical Shapley--Scarf market that we consider where preferences are strict). Both welfare measures are relevant in different real-world settings and therefore we think of these two research strands as complementary.
	
	Our work is also related to a series of recent articles that have studied the integration of other types of markets without money, in particular for Gale--Shapley one-to-one matching markets \citep{ortega2018social,ortega2019losses}, Gale--Shapley many-to-one matching markets with applications to school choice \citep{manjunath2016,dougan2019unified,ekmekci2019common,turhan2019,aue2020}, exchange economies \citep{chambers2017gains,chambers2019can} and networking markets \citep{gersbach2019gainers}. Among these, the closest to ours are \cite{ortega2018social, ortega2019losses}. He shows that, in Gale--Shapley marriage markets, market integration never harms more agents than it benefits, even though the average rank of an agent’s spouse can decrease by 37.5\% of the length of agents' preference list. He also provides an approximation for the gains from integration in random markets. Some of our results parallel his for Gale--Shapley marriage markets, although ours are more general as: i) they apply to the integration of markets of different sizes, ii) they provide tight bounds on the welfare losses, and iii) in the case of the gains from integration in random markets, our results are exact rather than approximations.

	Our average-case results rely on two seminal papers from the computer science literature regarding random Shapley--Scarf markets with uniform and independent preferences. The first of these, by \cite{frieze1995probabilistic}, computes the expected number of iterations that the top trading cycles algorithm takes to find the unique core allocation and the number of cycles created in the process. The second paper, by \cite{knuth1996exact}, finds the expected sum of ranks of obtained houses and establishes the equivalence between the core allocation obtained from random endowments and the random serial dictatorship mechanism with no property rights.\footnote{The latter result was also independently discovered by \cite{abdulkadirouglu1998random}.} \cite{che2019efficiency} use a similar random market approach to show that, in a related two-sided model, the top trading cycles algorithm achieves efficiency and stability asymptotically when agents' preferences are independent. 
	
	\section{Model}
	\label{sec:model}
	\paragraph{Preliminary definitions} We study the housing market proposed by \cite{shapley1974}, where there are $n$ agents, each of them owning an indivisible good (say a house). The agents have strict ordinal preferences over all houses, including their own, and no agent has any use for more than one house.\footnote{We only consider the case where agents have strict preferences; for an analysis of housing markets with weak preferences, see \cite{quint2004houseswapping,alcalde2011exchange,aziz2012housing,jaramillo2012difference,saban2013house} and \cite{aslan2020competitive}.}
	
	Formally, let $N \coloneqq \{1,\ldots, n\}$ be the set of agents and let $\omega \coloneqq \{\omega_1,\ldots,\omega_n\}$ be the initial endowment of the market.
	Let $\succ_i$ denote the strict preference of agent $i$ and let  $\succ \coloneqq (\succ_i)_{i \in N}$. The weak preference corresponding to $\succ_i$ is denoted by $\succcurlyeq_i$. A \emph{housing market (HM)} is a pair $(N,\succ)$. An \emph{allocation} $x=\{x_1,\ldots, x_n\}$ is any permutation of the initial endowment. That is, $\omega_i$ (resp. $x_i$) denotes the house endowed (resp. allocated) to agent $i$.
	
	An allocation $x$ is \emph{individually rational} if $x_i \succcurlyeq_i \omega_i$ for all $i \in N$. An allocation $x$ is a \emph{core allocation} if there does not exist a coalition $S \subseteq N$ and an allocation $y$ such that $\{y_i : i \in S\}=\{\omega_i : i \in S\}$ and $y_i \succ_i x_i$ for all $i \in S$. 
	An allocation $x$ is \emph{Pareto optimal} if, for every alternative allocation $x'$ such that $x'_i \succ x_i$ for some $i \in N$, there exists some $j \in N$ for which $x_j \succ_j x'_j$. A \emph{matching mechanism} $\MM$ is a map from HMs to allocations, and is said to be a core one (resp. individually rational, Pareto optimal) if it produces a core (resp. individually rational, Pareto optimal) allocation for every HM.  The mechanism $\MM$ is \emph{strategy-proof} if, for every $i, \succ_i', \succ$, $\MM_i(N, \succ) \succcurlyeq_i \MM_i(N, (\succ_i',\succ_{-i}))$.
	
	There is a unique core allocation (henceforth denoted by $x^*$) in every housing market. The unique core allocation can be found with an algorithm known as \emph{top trading cycles} (TTC) \citep{shapley1974,roth1977}, which works by repeating the following two steps until all agents have been assigned a house.
	
	\begin{enumerate}
		\item Construct a graph with one vertex per agent.  Each agent points to the owner of his top-ranked house among the remaining ones. At least one cycle exists and no two cycles overlap. Select the cycles in this graph.
		
		\item Permanently assign to each agent in a cycle the object owned by the agent he points to. Remove all agents and objects involved in a cycle from the problem.
	\end{enumerate}
	
	TTC is the only mechanism satisfying individual rationality, Pareto-efficiency and strategy-proofness on the strict preference domain \citep{ma1994strategy}.
	
	\paragraph{New definitions} 
	
	We study \emph{extended housing markets} (EHM), which consist of a HM and a partition of the set of agents into $k$ disjoint communities $C_1, \ldots, C_k$. That is, an EHM is a triple ($N, \succ,C$), where $C\coloneqq \{C_1,\ldots, C_k\}$. An \emph{integrated allocation} is any allocation for the HM $(N,\succ)$, whereas a \emph{segregated allocation} is an allocation for $(N,\succ)$ in which every agent receives a house owned by an agent in her own community. That is, a segregated allocation $x$ is such that $\{x_i : i \in S \}= \{\omega_i : i \in S\}$ $\forall S \in C$. A matching scheme $\sigma$ is a map from EHMs into a segregated and integrated allocation, denoted by $\sigma(\cdot, C)$ and $\sigma(\cdot, N)$, respectively.\footnote{Matching schemes are similar to the concept of assignment schemes in cooperative game theory \citep{sprumont1990population}.}

	For agent's $i \in C_j$ preference $\succ_i$, we denote its restriction to $C_j$ by $\widetilde \succ_i$. In other words, $\widetilde \succ_i$ is the strict ranking of agent $i$ on all the houses belonging to agents in community $C_j$ (including his own) that is consistent with $\succ_i$.  The matching scheme $\sigma^*$ is the \emph{core matching scheme} if $\sigma^*(\cdot,N)$ is the core matching for the HM $(N,\succ)$ and, for every community $C_j$, $\sigma^*(\cdot,C_j)$ is the core matching for the HM $(C_j, \widetilde \succ_{C_j})$, where $\widetilde \succ_{C_j} \coloneqq (\widetilde \succ_i)_{i \in C_j}$. 
	
	The \emph{rank} of house $\omega_h$ in the preference order of agent $i$ is defined by $\rk_i(\omega_h) \coloneqq \left\vert \{ j \in N : \omega_j \succcurlyeq_i \omega_h \} \right\vert$. The \emph{gains from integration} for agent $i$ under the matching scheme $\sigma$ are defined as $\gamma_i (\sigma) \coloneqq  \rk_i(\sigma(i,C)) -  \rk_i(\sigma(i,N))$. The \emph{total gains from integration} are given by $\Gamma (\sigma) \coloneqq \sum_{i \in N} \gamma_i$. If these are negative, we speak of the total losses from integration. The \emph{average percentile gains from integration} are denoted by $\overline \Gamma (\sigma) \coloneqq \frac{\Gamma (\sigma)}{n^2}$. We divide by $n^2$ to account for both the number of agents ($n$) and the length of an agent's preference list (which is also $n$). Thus, $\overline \Gamma (\sigma) \in (-1,1)$, where $\overline \Gamma (\sigma)=-1$ means that everybody was harmed by integration and moved from their best possible house to the worst possible one. 
	
	We use $N^+(\sigma) \coloneqq \{ i \in N : \sigma(i, N) \succ_i \sigma(i, C)  \}$ to denote the set of agents who benefit from integration. Similarly, $N^0(\sigma) \coloneqq \{ i \in N : \sigma(i, N) = \sigma(i, C)  \}$ and $N^-(\sigma) \coloneqq \{ i \in N : \sigma(i, C) \succ_i \sigma(i, N)  \}$ denote the set of agents that are unaffected and harmed by integration, respectively. For all $j \in \{1, \ldots, k\}$, we define $N^+_{C_j}(\sigma) \coloneqq \{ i \in C_j : \sigma(i, N) \succ_i \sigma(i, C)  \}$ to be the set of agents in community $C_j$ who benefit from integration. The sets $N^0_{C_j}(\sigma)$ and $N^-_{C_j}(\sigma)$ are defined analogously.
	
	Henceforth we focus on $\sigma^*$, i.e. we study the gains from integration that occur when the allocation obtained before and after integration occurs is the unique core allocation.
	
	\section{Running Example}
	\label{sec:example}
	
	Example \ref{tab:ex1} presents an EHM that we will use throughout the paper to illustrate how market integration may harm the majority of agents, and how their welfare losses can be significant. In this EHM, $n=7$ and $k=2$ with $C_1=\{a,b,c\}$ and $C_2=\{d,e,f,g\}$. The integrated (resp. segregated) core allocation appears in a diamond (resp. circle). 	
	\begin{table}[!htbp]
		\begin{center}
			\caption[caption,justification=centering]{An EHM with $C_1=\{a,b,c\}$ and $C_2=\{d,e,f,g\}$.}
			\label{tab:ex1}
			\resizebox{.75\textwidth}{!}{%
				
				\begin{tabular}{ccc|cccc}
					
					a 				& b 			& c 			& d 			& e 			& f 			& g \\ 
					\midrule
					
					\ensquare{d} 	& \encircle{a} 	& \encircle{b} 	& \ensquare{a} 	& \encircle{d} 	& \encircle{e} 	& \encircle{f} \\ 
					\encircle{c}	&  d			& a  			& \encircle{g} 	&a 				& a				& a\\ 
					\vdots			& \ensquare{b}	& d				& \vdots 		&b  			& d 			&  d \\ 
					& \vdots		& \ensquare{c}	&  				&c  			& b 			& b \\ 
					&   			&  				&  				&\ensquare{e}  	& c 			& c\\ 
					& 				&  				&  				&\vdots  		&\ensquare{f}	& e\\ 
					&     			&  				&  				&  				& g				& \ensquare{g} \\ 
			\end{tabular} }
		\end{center}
	\end{table}
	
In Example \ref{tab:ex1}, there are two communities with three and four agents each, such that one agent from each community (in this case $a$ and $d$) is assigned to their second best house in the segregated core allocation, whereas all remaining agents are assigned to their most preferred house. However, when both communities integrate, $a$ and $d$ exchange their houses, each obtaining their most preferred house, and thus making that all other five agents are assigned to their own house, which they prefer less than the segregated core allocation.

The two agents who experience welfare gains ($a$ and $d$) go from their second to their first best after integration occurs, obtaining a rank gain of +1. However, agent $c$ goes from his first to his third best (a change of -2 in rank), agent $g$ goes from his first to his fourth best (a change of -3 in rank), and so on, until agent $e$ who goes from his best to his worst option (a change of -6 in rank). When we add the total welfare losses ($+1+1-2-3-4-5-6$), we obtain $-\frac{1}{2} (n^2 -n -k^2 - k)=-18$. Dividing -18 by $n^2=49$, we find an average welfare reduction of 36.7\% of the length of agents' preferences. 
	
	In the next section, we generalize these findings, providing upper bounds for i) the number of agents harmed by integration, and ii) the size of average welfare losses.
	\section{Worst-case Results}
	\label{sec:worst}
	
	Unfortunately, the integration of housing markets may harm the vast majority of agents. In the worst-case scenario, up to $n-k$ agents are harmed by integration, and this upper bound is tight.
	
	\begin{proposition}
		\label{thm:worst-number}
		For any pair $(n,k)$, there exists an EHM in which $|N^-(\sigma^*)|=n-k$; whereas there is no EHM in which $|N^-(\sigma^*)|> n-k$.
	\end{proposition}
	
	\begin{proof} 
		The EHM in Example \ref{tab:ex1} illustrates an EHM showing that the $n-k$ bound is attainable. We can extend the construction of this example to arbitrary values of $n$ and $k$ as follows:		
		\begin{enumerate}
			\item Enumerate agents arbitrarily so that agents from community $C_1$ are first, then those in $C_2$, and so on. Separate agents into two sets, namely $X$ and $N \setminus X$. The set $X$ contains the first agent from each community only. The agent from community $C_j$ in $X$ is denoted by $j^*$. The last agent in community $C_j$ (which is in $N \setminus X$) is denoted by $\underline j$.
			
			\item The preferences for any agent $i^* \in X$ are such that: 
			\begin{enumerate}
				\item $\rk_{\succ_{i^*}}((i+1)^*)=1$ (modulo $k$) and 
				\item $\rk_{\succ_{i^*}}(\underline i)=2.$
			\end{enumerate}

			\item The preferences for any agent $i \in N \setminus X$ are such that: 
			
			\begin{enumerate}
				\item $\rk_{\succ_{i}}(i-1)=1$,
				\item For any $j^* \in X$ and $h \in N \setminus X$ (with $h \neq i+1$),  $\rk_{\succ_{i}}(j^*)< \rk_{\succ_{i}}(h)$, and
				\item For any two $h, h' \in N\setminus X$ and $h, h' \neq i+1$, $\rk_{\succ_{i}}(h)< \rk_{\succ_{i}}(h')$ if $h < h'$.
			\end{enumerate}
		\end{enumerate} 
		
		Constructing the preferences in such a way guarantees that, in the segregated core allocation, every agent in $N \setminus X$ obtains their first choice, whereas every agent in $X$ gets their second choice. In contrast, in the integrated core allocation, every agent in $X$ obtains their first choice, whereas everybody in $N \setminus X$ obtains an object ranked from $k+1$ to $n$. Example \ref{tab:ex1} was constructed in this fashion.
		
		To see that the $n-k$ upper bound is tight, assume by contradiction that more than $n-k$ agents are harmed by integration, which implies that there is one community in which all agents are harmed by integration, say $C_j$. But then $\sigma^*(\cdot, N)$ is not a core allocation for $(N, \succ)$, because any alternative allocation $x$ such that $x_i=\sigma^*(i,C)\; \forall i \in C_j$ dominates it (since $C_j$ is effective for allocation $x$ and every agent in $C_j$ prefers the segregated over the integrated allocation). That the integrated core allocation is not a core allocation is a contradiction, which terminates the proof.
	\end{proof}
	
	Proposition \ref{thm:worst-number} implies that integration may harm the majority of agents in Shapley--Scarf markets. This is a striking observation, since the integration of Gale--Shapley marriage markets (in which two sets of agents are matched to each other) always benefits more agents than those it harms (see Proposition 2 in \cite{ortega2018social}, also \cite{gale1962college,gardenfors1975match}).\footnote{One may consider the opposite scenario, in which an integrated market of size $n$ breaks into $k$ disjoint communities. Simple examples shows that all agents can become worse off after markets disintegrate, irrespective of the value of $k$.}
	
	Given the negative result in Proposition \ref{thm:worst-number}, we may think that integration can still be justified if the size of the welfare gains experienced by a minority are much larger than the size of the welfare losses suffered by a majority. Unfortunately, in the worst-case scenario, the size of the losses from integration is much larger than the size of the gains from integration. In particular, we show below that the agents' average welfare loss may be negative and asymptotically equivalent to an increase in ranking of 50\% of the length of agents' preference list. We provide a tight lower bound on the size of agents' average welfare loss.
	
	\begin{proposition}
		\label{thm:worst-size}
		For any pair ($n,k$), there exists an EHM in which $\overline \Gamma (\sigma^*) = \frac{-n^2+n+k^2+k}{2 n^2}$; whereas there is no EHM in which $\overline \Gamma (\sigma^*) < \frac{-n^2+n+k^2+k}{2 n^2}$
	\end{proposition}
	
	\begin{proof}
		Example \ref{tab:ex1} shows that our lower bound for $\overline \Gamma (\sigma^*)$ is attainable. We constructed the EHM in Example \ref{tab:ex1} in such a way that the minimum possible number of agents gain from integration (i.e. $k$, per Proposition \ref{thm:worst-number}), and that the size of such gains is as small as possible (+1). On the other side, the welfare losses of the remaining $n-k$ individuals go from $-2$ to $-n+1$ (the largest possible welfare loss). We can replicate such construction for EHMs with arbitrary values of $n$ and $k$ as described in the proof of Proposition \ref{thm:worst-number} to obtain:
		\begin{eqnarray}
			\overline \Gamma (\sigma^*) &=& \frac{1}{n^2} \left(k*1- \sum_{i=1}^{n-k} n-i \right)\\
			&=& \frac{1}{n^2}\left(k-n(n-k)+\sum_{i=1}^{n-k} i\right)\\
			&=& \frac{1}{n^2} \left(k - n^2 +nk+\frac{(n-k)(n-k+1)}{2}\right)\\
			&=& -\frac{1}{2n^2} (n^2 -n -k^2 - k)
		\end{eqnarray}
		
		This establishes that our lower bound can be attained for arbitrary values of $n$ and $k$. It is interesting that our lower bound does not depend on the size of each community relative to the size of the whole society. Note that when $n$ grows and $k$ remains constant, $\overline \Gamma (\sigma^*) \sim -1/2$. 
		
		We now show that our lower bound for $\overline \Gamma (\sigma^*)$ is tight, with the help of some additional definitions and two auxiliary lemmas. Given a core allocation $x^*$ for a HM  $(N, \succ)$  and an integer $r$ such that $1\leq r \leq n$, let $m(r,x^*) \coloneqq | \{i \in N: \rk_i (x_i^*)\} =r|$. Similarly, let $M(r,x^*) \coloneqq | \{i \in N: \rk_i (x_i^*)\} \geq r|$.

		\begin{lemma}
			\label{lemma1}
			In any core allocation $x^*$, $\rk_i(x^*_i)\leq \rk_i(\omega_i)$.
		\end{lemma}
		\begin{proof}
			This is a well-known fact due to any core allocation being individually rational.
		\end{proof}
		
		\begin{lemma}
			\label{lemma2}
			In any core allocation $x^*$, $m(r, x^*) \leq n-r+1$.
		\end{lemma}
		
		\begin{proof}
			
			For $r=n$, our lemma says $m(n,x^*)\leq 1$. Note that if $\rk_i(x_i^*)=n$, then ${x_i}^*=\omega_i$ because of Lemma \ref{lemma1}. Therefore, we cannot have $m(n,x^*)> 1$, as otherwise two agents are assigned their own house but they would like to exchange their house with each other, and thus $x^*$ is not a core allocation.

			For $r=n-1$, suppose by contradiction that $m(n-1, x^*) >2$. Then there exists three agents $j,l,h$ for which $\rk(x^*_i)=n-1$ for all $i \in \{j,l,h\}$. But for each of those agents, there exists a house $\omega'_i \in \{\omega_j,\omega_l,\omega_h\}$ such that $\omega'_i \succ_i x_i^*$ and $\omega'_i \succ_i \omega_i$ for all $i \in \{j,l,h\}$. Therefore, $x^*$ is not a core allocation, since there is a reallocation of houses among $j,l,h$ that is effective for such coalition and that is strictly preferred. 
			
			The same argument applies for any other values of $r<n-1$. Suppose by contradiction that there exist some $r' \le n-1$ such that $m(r', x^*) > n-r' +2$. Then there are $n-r'+2$ agents for which $\rk(x^*_i)=r'$. But for each of these agents $i$, there exists a house $\omega_j$ belonging to one of these $n-r'+2$ agents such that $\omega_j \succ_i x_i^*$ and $\omega_j \succ_i \omega_i$. Therefore, $x^*$ is not a core allocation, since there is a reallocation of houses among those $n-r'+2$ agents that is effective for such coalition and that is strictly preferred. Hence, the argument holds for all $r$.
			
		\end{proof}
		
		\begin{lemma}
			\label{lemma3}
			In any core allocation $x^*$, $M(r, x^*) \leq n-r+1$.
		\end{lemma}
		
		\begin{proof}
			For $r=n$, the statement in Lemma \ref{lemma3} is the same as in Lemma \ref{lemma2}. For $r=n-1$, assume by contradiction that $M(n-1,x^*)>2$. By Lemma 2 we cannot have that two agents are allocated a house ranked $n$ for both, or that three agents are allocated a house ranked $n-1$. Thus, it must be that one agent gets a house ranked $n$ (agent $j$) and two agents get a house ranked $n-1$ (agents $h$ and $l$). Then we have $x_j=\omega_j$ by Lemma \ref{lemma1}. Furthermore, for $i \in \{h,l\}$, there are two houses $x'_i, x_i'' \in \{\omega_j,\omega_h,\omega_l\}$ such that  $x'_i \succ_i x_i$ and $x_i \succ_i x''_i$, where $x'_i \neq \omega_i$ per Lemma \ref{lemma1}. If, for either agent $h$ or $l$, $x_i'=\omega_j$, then $j$ and such agent would like to exchange their endowments and would be strictly better off, and thus $\rk_h(\omega_j)=\rk_l(\omega_j)=n$. But because $\rk_h(x_h)=\rk_l(x_l)=n-1$, they must be getting their own houses, i.e. $x_h=\omega_h$ and $x_l=\omega_l$. But then, agents $h$ and $l$ are better of by trading their endowments, and thus $x^*$ is not a core allocation, a contradiction. The same argument applies for all other values of $r<n-1$.
		\end{proof}

		Armed with these three auxiliary lemmas, we are ready to prove that $\overline \Gamma (\sigma^*) \geq \frac{-n^2+n+k^2+k}{2 n^2}$. By Proposition \ref{thm:worst-number}, at most $n-k$ people may experience negative gains from integration. These are defined, for each agent $i$, as $\gamma_i(\sigma^*)\coloneqq \rk_i(\sigma^*_i(i,C))-\rk_i(\sigma^*_i(i,N))$. To make $\gamma_i(\sigma^*)$ as small as possible, we need to fix $\rk_i (\sigma^*_i(i,C))=1$ and make $\rk_i(\sigma^*_i(i,N))$ as large as possible. But Lemma 3 shows that $\rk(\sigma^*_i(i,N))=n$ for at most one agent, $\rk_i(\sigma^*_i(i,N)) \geq n-1$ for at most two agents, and so on. Thus, in the worst case scenario, the sum of the welfare gains from integration among those $n-k$ agents equals
		\begin{equation}
			-\sum_{i=1}^{n-k} (n-i)=\frac{-n^2+n+k^2-k}{2}
		\end{equation}
		
		Similarly, the smallest positive gains from integration for the remaining $k$ agents (which must exists by Proposition \ref{thm:worst-number}) are equal to 1. Thus, the smallest possible value for $\overline \Gamma (\sigma^*)$ is 
		\begin{equation}
			\overline \Gamma (\sigma^*) =-\frac{1}{2n^2} (n^2 -n -k^2 - k)
		\end{equation}
	\end{proof}
	
	Proposition \ref{thm:worst-size} can be compared to an analogous result in Gale--Shapley marriage markets. The average welfare gains may also be negative in Gale--Shapley markets, but only up to 37.5\% of the length of preference lists \citep{ortega2019losses}.\footnote{This lower bound is not proven to be tight but is the best bound available.} Taken together, Propositions \ref{thm:worst-number} and \ref{thm:worst-size} show that the integration of Shapley--Scarf markets can be hard to achieve, and in particular is more difficult to obtain (in the worst-case scenario) than in Gale--Shapley marriage markets.
	
	\section{Average-case Results}
	\label{sec:random}
	In the previous section we found two negative results regarding the integration of Shapley--Scarf markets; however both results are about worst-case scenarios. While these results are interesting on their own, one may argue that these are knife-edge scenarios, and wonder whether market integration would generate welfare gains on average.  
	
	To answer this question, we study \emph{random  housing markets} (RHM). Given a set of agents, a RHM is generated by drawing a complete preference list for each agent independently and uniformly at random. Similarly, a \emph{random extended housing market} (REHM) is a RHM where the set of agents is partitioned into disjoint communities $C_1, \ldots, C_k$, each of size $n_1, \ldots, n_k$ (where $n=n_1+\ldots+n_k$). 
	We emphasize that the randomness refers to agents' preferences and not to the partition $C$, which is deterministic.
	Random housing markets were first studied by \cite{frieze1995probabilistic} and \cite{knuth1996exact}. The latter proved the following seminal result.
	
	\begin{lemma}[\citealp{knuth1996exact}]
		In a RHM, $\EE(\sum_{i=1}^n \rk_i(x_i^*)) = (n+1)H_n-n$, where $H_n$ is the $n$-th harmonic number, i.e. $H_n \coloneqq \sum_{i=1}^n \frac{1}{i}$.
	\end{lemma} 
	
	We can use Knuth's theorem to find the expected size of the average welfare gains in REHMs. Let us define the \emph{total gains from integration for community $C_j$} as $\Gamma_{C_j} (\sigma) \coloneqq \sum_{i \in C_j} \gamma_i$.\footnote{Recall that $\gamma_i (\sigma) \coloneqq  \rk_i(\sigma(i,C)) -  \rk_i(\sigma(i,N))$.} The \emph{average percentile gains from integration for community $C_j$} are denoted by $\overline \Gamma _{C_j}(\sigma) \coloneqq \frac{\Gamma (\sigma)}{n\, n_j}$. We divide by $n_j$ to take the average across all agents in community $C_j$, and by $n$ to normalize by the length of agents' preference lists. Equipped with these new definitions, we can compute average gains from integration, which are positive for agents belonging to any community. 
	
	\begin{proposition}
		\label{thm:avg-size}
		$\EE [\overline \Gamma_{C_j} (\sigma^*)] =   \frac{(n+1)[ (n_j+1)H_{n_j}-n_j]}{n_j(n_j+1)n}-\frac{(n+1)H_n - n}{n^2}$. 
	\end{proposition}
	
	\begin{proof}
		For any $i, j  \in C_j$ and any community $C_j$, define the \emph{relative rank} of house $\omega_h$ in the preference order of agent $i$ by $\hat \rk_i(\omega_h) \coloneqq \left\vert \{ l \in C_j : \omega_l \succcurlyeq_i \omega_h \} \right\vert$. This is, the relative rank indicates the position of a house in an agent's preference ranking compared \emph{only} to houses owned by other agents belonging to the same community.
		Knuth's result directly implies that
		\begin{eqnarray}
			\EE [\sum_{i=1}^n \rk_i(\sigma^*(i,N))] &=& (n+1)H_n-n, \text{and}\\
			\label{eq:relative}	\EE[\sum_{i=1}^{n_j} \hat \rk_i(\sigma^*(i,C_j))] &=& (n_j+1)H_{n_j}-n_j, \, \forall j \in \{1, \ldots, k\}
		\end{eqnarray} 
		
		So that before integration,	agents are assigned to a house relatively ranked $(n_j+1)H_{n_j}-n_j$. To complete the proof, we need to figure out in which position is such house in the absolute rank of all houses (i.e. convert the relative rank into the full rank). To do so, suppose that a house assigned to an agent in a segregated allocation has a relative rank $q$. A randomly chosen house, belonging to an agent from another community, could be better ranked than house 1, between houses 1 and 2, ..., between
		houses $q-1$ and $q$, and so on. Therefore, a random house belonging to another
		agent is in any of those gaps with probability $\frac{1}{n_j+1}$ and thus has
		$\frac{q}{n_j + 1}$ chances of being more highly ranked than the house with relative ranking $q$. There are $(n-n_j)$ houses from other communities.
		On average, $\frac{q(n-n_j)}{n_j+1}$ houses will be better ranked. Furthermore, there were
		already $q$ houses in his own community ranked better than it. This implies
		that its expected ranking is $q+\frac{q(n-n_j)}{n_j+1}=\frac{q(n+1)}{n_j+1}$. Substituting $q$ for the expression obtained in equation \eqref{eq:relative}, we obtain
		\begin{eqnarray}
			\EE [\overline \Gamma_{C_j} (\sigma^*))] =  \frac{(n+1)[ (n_j+1)H_{n_j}-n_j]}{n_j(n_j+1)n}&-&\frac{(n+1)H_n - n}{n^2}
		\end{eqnarray}
	\end{proof}

	Proposition \ref{thm:avg-size}, which is interesting per se, provides valuable comparative statics, which we present in the following Corollary.
	
	\begin{corollary}
		The expected welfare gains from integration are positive for all agents, and higher for agents in smaller communities.
	\end{corollary}
	
	For example, if we merge three Shapley--Scarf markets of size 60, 30 and 10, the corresponding welfare gains in terms of house rank are 1.98, 5.95 and 16.16, i.e. agents from the market with only 10 agents improve the ranking of their assigned house by 16 positions, whereas those in the market with 60 agents only improve theirs by 2 positions. In percentile terms, agents from the smallest market improve the rank of their assigned house by 16\% of the length of their preference list, whereas agents from the largest market increase their corresponding rank only by 2\% of the length of their preference list. Our theoretical predictions match very accurately the gains from integration observed in simulated random markets. Averaging the results of a thousand random markets (with three markets each of sizes 60, 30 and 10), we obtain that the realized gains from integration are of 2.07, 5.93 and 16.12 (with standard deviations of 1.31, 2.27 and 5.82, respectively).\footnote{The corresponding code is available from \url{www.josueortega.com}.}
	
	We now turn to studying the expected number of agents who are harmed by integration in each community, i.e. $|N^-_{C_j}(\sigma^*)|$. To do so, we relate the number of trading cycles in TTC for the segregated markets to the number of agents harmed by integration via two auxiliary Lemmas. For any community $C_j$, let $t_j$ be the number of cycles obtained by TTC when computing the segregated core allocation $\sigma^*(\cdot,C_j)$, and let $t \coloneqq \sum_{j=1}^k t_j$.\footnote{To clarify, $t$ is the number of cycles, not of iterations. One iteration in TTC may generate more than one cycle.} 
		
		Our first auxiliary Lemma relates $|N^-_{C_j}(\sigma^*) |$ to $t_j$.
		\begin{lemma}
			\label{thm:tj}
			In any EHM, $|N^-_{C_j}(\sigma^*) | \leq n_j-t_j$.
		\end{lemma}
		
		\begin{proof}
			In any cycle obtained by TTC when computing the segregated core allocation $\sigma^*(\cdot,C_j)$, we must either have that all agents in the cycle are in $N^0_{C_j}(\sigma^*)$ or that at least one agent is in $N^+_{C_j}(\sigma^*)$. Otherwise there is a cycle (involving a set of agents $S$) with at least one agent in $N^-_{C_j}(\sigma^*)$ and with no agent in $N^+_{C_j}(\sigma^*)$. Such a combination cannot occur. If all agents in the cycle are in $N^-_{C_j}(\sigma^*)$, then those agents are clearly a blocking coalition to the integrated core allocation. If some agents are in $N^-_{C_j}(\sigma^*)$ and some in $N^0_{C_j}(\sigma^*)$, then when we run TTC to find the integrated core allocation, there is an agent $i \in N^-_{C_j}(\sigma^*)$ who is pointed by an agent $h \in N^0_{C_j}(\sigma^*)$, i.e. $h$'s assignment does not change (it is $\omega_i$ before and after integration) but the one of $i$ becomes worse. But when we run TTC, $i$ points to the agent owning the best house available. Now, if $\sigma^*(i,C)$ is no longer available, it means that its owner exited in an earlier cycle during TTC, and thus she must have received a better house, and thus there is an agent in $N^+_{C_j}(\sigma^*)$, a contradiction.
		\end{proof}
		
		Our second auxiliary lemma computes the expected number of cycles in random housing markets. It appears as Theorem 2 in \cite{frieze1995probabilistic}. Let $t'$ denote the number of cycles formed during the execution of TTC in a RHM with $n'$ agents. Then,
		
		\begin{lemma}[\citealp{frieze1995probabilistic}]
			\label{lemma:frieze}
			$\EE[t']=\sqrt{2 \pi n'}+ O(\log n')$.
		\end{lemma} 
		
		Note that Lemma \ref{lemma:frieze} implies that, in a REHM:
		\begin{eqnarray}
			\EE[t_j]&=&\sqrt{2 \pi n_j}+ O(\log n_j) \label{eqn:frieze}
		\end{eqnarray}
		
		Combining Lemmas \ref{thm:tj} and \ref{lemma:frieze}, we obtain an upper bound on the expected number of agents harmed by integration in each community. Proposition \ref{thm:avg-number}	 below presents this upper bound.
		
			\begin{proposition}
			\label{thm:avg-number}
			$\EE[|N^-_{C_j}(\sigma^*)|] \leq n_j-\sqrt{2\pi n_j}  - O(\log  n_j) 	$.
		\end{proposition}
		
		\begin{proof}
		Substituting $t_j$ in Lemma \ref{lemma:frieze} for its value in equation \eqref{eqn:frieze} together, we directly obtain the proof of our result. 
	\end{proof}
	
	Proposition \ref{thm:avg-number} provides, as a Corollary, a bound on the expected total number of agents in the whole economy that are harmed by integration.
	
	\begin{corollary}
		$\EE[N^-(\sigma^*)] \leq n-\sqrt{2\pi} (\sum_{j=1}^k \sqrt{n_j}) - O(\log \prod_{j=1}^k n_j)$.
	\end{corollary}
	
	Proposition \ref{thm:avg-number} is our only bound that is not tight, but is nevertheless informative. Returning to our example of a EHM divided into three communities of sizes 60, 30 and 10, Proposition \ref{thm:avg-number} tells us that, on average, the TTC algorithm generates around 30 trading cycles when computing the integrated core allocation. In each of those cycles, at least one person is not harmed by integration. Consequently, at most 70 agents can be harmed by integration. But in fact Proposition \ref{thm:avg-number} says more: it tells us the distribution of agents harmed by integration across communities. Thus, in the market of size 60, the expected number of agents harmed by integration is smaller than 44. Similarly, for the markets of size 30 and 10, the expected number of agents harmed by integration is smaller than 19 and 4, respectively. 
	
	Another Corollary of Proposition \ref{thm:avg-number} is that, whenever all communities have the same number of agents $n_1$, market integration never harms more than half of the total population if $n_1$ is sufficiently small. 
	
	\begin{corollary}
		If $n_1=\ldots=n_k$, then $\EE[|N^-(\sigma^*)|] \leq \frac{n}{2}$ if $n_1\leq 8\pi \approx 25.13$.
	\end{corollary}
	
	\begin{proof} From Corollary 2, we have that:
		\begin{eqnarray}
			\EE[N^-(\sigma^*)] &=&kn_1 - k\sqrt{2\pi n_1}- O(\log n_1^k)
		\end{eqnarray}
		
		and therefore $\EE[N^-(\sigma^*)]$ is less than $n/2$ when
		\begin{eqnarray}
			k (n_1 - \sqrt{2 \pi n_1}) - O(\log n_1^k)&\leq& \frac{kn_1}{2}\\
			n_1 &\leq& 2[\sqrt{2 \pi n_1} + O(\log n_1^k)/k]
		\end{eqnarray}
		In particular, condition (13) is satisfied whenever:
		\begin{eqnarray}
			n_1 &\leq& 2[\sqrt{2 \pi n_1}] = 8\pi \approx 25.13
		\end{eqnarray}
	\end{proof}
	
	For several sensible combinations of parameters, we never observe that the number of agents harmed by integration was over 50\%. The fraction of agents harmed by integration was between 14\% to 22\%, and becomes smaller as $k$ increases and as $n$ decreases (see Table \ref{tab:expranking}). The intuition behind these changes is that as $k$ grows, integrations offers more opportunities for trade; whereas when $n$ grows (and $k$ remains constant) the probability that the integrated and segregated matchings are the same becomes smaller, and therefore more agents benefit and are harmed by market integration (because fewer agents are unaffected by integration). Our simulations shows that the bound in Proposition \ref{thm:avg-number} regarding the number of agents harmed by integration can be improved. We leave this interesting question for future research.
	\renewcommand{\tablename}{Table} 
	\begin{table}[ht]
		\centering
		\caption[caption,justification=centering]{Fraction of agents affected by integration. \hspace{\textwidth}\scriptsize	 Average over a thousand simulations with preferences drawn uniformly at random. Standard errors in parenthesis.}
		\label{tab:expranking}
		
		\begin{tabular}{lrrrrrr}
			\toprule
			$n$&		\multicolumn{6}{c}{$k$}\\
			&	\multicolumn{2}{c}{2}   & 	\multicolumn{2}{c}{3}    & 	\multicolumn{2}{c}{5}        \\
			\cmidrule(r){2-3} \cmidrule(r){4-5} \cmidrule(r){6-7}
			& Benefit&Harmed    & Benefit&Harmed    & Benefit&Harmed    \\
			\cmidrule(r){2-3} \cmidrule(r){4-5} \cmidrule(r){6-7}
			25        & 53.52 & 19.95 & 64.68 & 17.63 & 75.05     & 14.08     \\
			& \footnotesize (6.8)   & \footnotesize (4.66)  & \footnotesize (5)     & \footnotesize (3.54)  & \footnotesize (3.59)      & \footnotesize (2.55)      \\
			50        & 54.64 & 21.45 & 65.57 & 18.47 & 75.54     & 14.67     \\
			& \footnotesize (4.47)  & \footnotesize (3.47)  & \footnotesize (3.52)  & \footnotesize (2.66)  & \footnotesize (2.41)      & \footnotesize (1.84)      \\
			100       & 55.4  & 22.17 & 66.06 & 18.99 & 75.88 & 14.88 \\
			\textbf{} & \footnotesize (3.28)  & \footnotesize (2.34)  & \footnotesize (2.42)  & \footnotesize (1.8)   &  \footnotesize (1.78)         & \footnotesize (1.34)\\
			\bottomrule
		\end{tabular}
	\end{table}
	
	\paragraph{Other Interpretations} As discussed in the related literature section, the core from random endowments is equivalent to the allocation obtained with random serial dictatorship in a market with no property rights, i.e. assigning a random order among agents and letting them choose their most preferred object that remains available according to such order \citep{knuth1996exact,abdulkadirouglu1998random}. Therefore, the results obtained for random markets in this section also apply to the integration of markets with no endowments in which random serial dictatorship is used.

	\section{Specific Preference Domains}
	\label{sec:domains}
	
	Although uniform and independent preferences are the most natural and simple preferences to consider in random markets, it is well-known that in real-life applications such as kidney exchange, agents' preferences are strongly correlated, with some ``houses'' being particularly desired by most agents. In this section, we show that if preferences satisfy a particular type of correlation structure, we can guarantee that no more than half of the total population of agents is harmed by integration.
	
	To do so, let $q(r,\widetilde \succ_{C_j})$ be the set of agents in community $C_j$ placed at rank $r$ by any agent in their own community (including themselves) in preference profile $\widetilde \succ_{C_j}$. This is, for any positive integer $r$ and any $j \in \{1, \ldots, k\}$, $q(r,\widetilde \succ_{C_j}) \coloneqq \{ i \in C_j : \exists h \in C_j : \rk_j(\omega_i)=r \}$. Similarly, let $Q(r,\widetilde \succ_{C_j}) \coloneqq \bigcup\limits_{t=1}^{r} {q(t,\widetilde \succ_{C_j})}$ be the set of agents in community $j$ placed at rank $r$ and above. 
	
	Now we introduce the property that will ensure that market integration does not harm a majority of agents, which we call \emph{sequential dual dictator}. This property was recently introduced by \cite{troyan2019obviously} in a two-sided extension of a Shapley--Scarf market, which he used to characterize the obvious strategy-proof implementation of TTC.
	
	\begin{definition}[Sequential dual dictator property]
		A preference profile $\succ$ satisfies the sequential dual dictator property if, for any positive integer $r$ and $\forall j \in \{1, \ldots, k\}$, each of their corresponding preference restriction $\succ_{C_j}$ satisfies
		$$  |Q(r,\widetilde \succ_{C_j})|  \leq r+1 $$
	\end{definition} 
	
	In Example \ref{tab:ex2}, we show that the preference profile in Example \ref{tab:ex1} does not satisfy the sequential dual dictator property and provide a preference profile that does. In Example \ref{tab:ex1}, $|Q(1,\widetilde \succ_{C_1})|=|\{b,c,a\}|>2$, violating the sequential  dual dictator property. Similarly, $|Q(1,\widetilde \succ_{C_2})|=|\{e,f,g,d\}|>2$. In contrast, in the profile on the right in Example \ref{tab:ex2},  $|Q(1,\widetilde \succ_{C_1})|=|\{c,a\}|\leq 2$, $|Q(1,\widetilde \succ_{C_2})|=|\{e,f\}| \leq 2$ and $|Q(2,\widetilde \succ_{C_2})|=|\{e,f,d\}|\leq 3$. Whenever preferences satisfy the sequential dual dictator property, we can guarantee that no more than half of the agents in each community are harmed by integration. Note that, in contrast to Proposition \ref{thm:avg-number}, here we bound the number of agents harmed by integration in every EHM, instead of the expected number of agents harmed by integration across all REHMs.
	\renewcommand{\tablename}{Example}
	\begin{table}[!htbp]
		
		\begin{center}
			\caption[caption,justification=centering]{The preference profile on the right satisfies the sequential dual dictator property, unlike the one on the left.}
			\label{tab:ex2}
			\parbox{.45\linewidth}{
				\centering
				\begin{tabular}{ccc|cccc}
					
					a & b & c & d & e & f & g \\ 
					\hline
					b & c & a 	& e & f 	& g 	& d \\ 
					\vdots &  	\vdots		&  	\vdots		& \vdots & 		\vdots		& 		\vdots		& \vdots \\ 
					&  			&  	&  				&  &  &   \\ 
					&  			&  &  &  &  &  \\ 
				\end{tabular} 
			}\hfill
			\parbox{.45\linewidth}{
				\centering
				\begin{tabular}{ccc|cccc}
					a & b & c & d & e & f & g \\ 
					\hline
					c & c & a & e & f & f & e \\ 			 
					b & b & c  & f & d & e & d \\
					a & c & b & d & g & g & f  \\ 
					&  &   & g & e & d & g  \\ 
				\end{tabular} 
			}
		\end{center}
	\end{table}

	\begin{proposition}
		\label{thm:domain}
		If $\succ$ satisfies the sequential dual dictator property, then $|N^-_{C_j}(\sigma^*)|\leq \frac{n_j}{2}$.
	\end{proposition}
	
	\begin{proof}
		To complete the proof, we examine the number and length of trading cycles generated by the TTC algorithm when computing the segregated core allocation $\sigma^*(\cdot, C_j)$ for community $C_j$. At the first iteration, all agents point to the owner of their most preferred house, and if the sequential dual dictator property is satisfied, there are only two vertices with a positive in-degree. A trading cycle is created, either of those agents pointing to themselves or pointing at each other, and therefore each cycle created in the first iteration of TTC has length at most 2. In the second iteration, at most two agents have positive in-degree (because at least one agent was removed in the first iteration). Either one or two cycles are formed in iteration 2, and they have length of at most 2. The argument repeats for each iteration: each trading cycle has length at most 2.
		
		Now we invoke an argument that we used in the proof of Lemma 5, showing that in any cycle, we must either have that all agents are in $N^0(\sigma^*)$ or that at least one agent is in $N^+_{C_j}(\sigma^*)$. We have showed that there are at least $n_j/2$ cycles in each community. Therefore, $|N^-_{C_j}(\sigma^*)|\leq \frac{n_j}{2}$. 
	\end{proof}
	
	One particular case of preference profiles satisfying the sequential dual dictator property are those in which all agents have the same preferences. Such preferences has been extensively studied in Gale--Shapley marriage markets because they guarantee the uniqueness of the core allocation and ensure that truth-telling is a Nash equilibrium of the revelation game induced by any stable mechanism \citep{gusfield1989stable}. The sequential dual dictator preference domain is  larger than this classical domain of equal preferences. The sequential dual dictatorship only imposes a particular structure on the preferences of each community over its own houses and is therefore substantially less restrictive than identical preferences.
	
	\section{Conclusion}
	\label{sec:concl}
	
Market integration leads to more efficient outcomes and yet real-life offers plenty of examples of markets that fail to integrate and operate disjointly. In this paper, we have provided results that shed light on why this might be the case for a specific type of markets in which monetary transfers are not permitted.
	
Our explanation lies in the fact that market integration may have negative consequences for most traders. These negative consequences are so dire that they vastly outweigh the welfare benefits of those who become better off with market integration. Somewhat surprisingly, the average effect on the economy can be so bad that the average trader ends up with an allocation in the lower half of their preference list.

These negative consequences of market integration, however, are the exception rather than the rule, as we have shown formally. Two interesting open problems for further research are: i) to obtain conditions that fully characterize which types of Shapley--Scarf markets benefit from integration, and more generally, ii) to provide a comprehensive and unified discussion of the institutional features that prevent markets from integrating.
	
	\section*{Acknowledgements}
	We are grateful to Debasis Mishra, Herv\'e Moulin, Arunava Sen and the anonymous referees of this journal and the 13th Symposium on Algorithmic Game Theory (where this paper appeared as a one-page abstract) for their helpful comments. Sarah Fox proofread this paper. This work was completed while Kriti Manocha was a visiting student at Queen's University Belfast; she gratefully acknowledges the University for their hospitality. 
	Josu\'e Ortega is partially supported by the UK Economic and Social Research Council, grant R1379QMs.
	Rajnish Kumar acknowledges financial support from British Council grant UGC-UKIERI 2016-17-059.
	\setlength{\bibsep}{0cm}
	\bibliographystyle{ecta}

\end{document}